\documentclass[11pt]{article}

\usepackage{amssymb,amsmath,amsfonts,amsthm,color}
\usepackage[margin=1in]{geometry}
\usepackage[pagebackref,letterpaper=true,colorlinks=true,pdfpagemode=none, linkcolor=blue,citecolor=blue,pdfstartview=FitH]{hyperref}
\usepackage{algorithm}
\usepackage{algorithmic}
\usepackage{xspace}

\theoremstyle{definition}
\newtheorem{theorem}{Theorem}

\newtheorem{prop}[theorem]{Proposition}
\newtheorem{coro}[theorem]{Corollary}

\def\eps{\epsilon}
\def\suchthat{\;:\;}

\def\indic{\mathbb{I}}

\newcommand{\abs}[1]{\left|#1\right|}
\newcommand{\size}[1]{\left|#1\right|}
\newcommand{\norm}[1]{\left\|#1\right\|}
\newcommand{\fnorm}[1]{\left\|#1\right\|_{F}}

\newcommand{\inner}[2]{\left\langle#1, #2\right\rangle}

\newcommand{\sparsestcut}{\textsc{Sparsest Cut }}
\newcommand{\uniformSparsestcut}{\textsc{Uniform Sparsest Cut }}

\title{Guruswami-Sinop Rounding without Higher Level Lasserre}
\author{Amit Deshpande\thanks{Microsoft Research, {\tt amitdesh@microsoft.com}} \and Rakesh Venkat\thanks{Tata Institute of Fundamental Research, Mumbai, {\tt rakesh@tifr.res.in}}}
\date{}

\begin{document}
\maketitle

\begin{abstract}
Guruswami and Sinop \cite{GuruswamiS13} give a $O(1/\delta)$ approximation guarantee for the non-uniform \sparsestcut problem by solving $O(r)$-level Lasserre semidefinite constraints, provided that the generalized eigenvalues of the Laplacians of the cost and demand graphs satisfy a certain spectral condition, namely, $\lambda_{r+1} \geq \Phi^{*}/(1-\delta)$. Their key idea is a rounding technique that first maps a vector-valued solution to $[0, 1]$ using appropriately scaled projections onto Lasserre vectors. In this paper, we show that similar projections and analysis can be obtained using only $\ell_{2}^{2}$ triangle inequality constraints. This results in a $O(r/\delta^{2})$ approximation guarantee for the non-uniform \sparsestcut problem by adding only $\ell_{2}^{2}$ triangle inequality constraints to the usual semidefinite program, provided that the same spectral condition $\lambda_{r+1} \geq \Phi^{*}/(1-\delta)$ holds as above.
\end{abstract}

\section{Introduction}
Finding sparse cuts in graphs or networks is a difficult theoretical problem with numerous practical applications, namely, divide-and-conquer graph algorithms, image segmentation \cite{ShiM00,SinopG07a}, VLSI layout \cite{BhattL84}, routing in distributed networks \cite{AwerbuchP90a}. From the theoretical side, the problem of finding the sparsest cut in a given graph is NP-hard, and over the years, significant efforts and non-trivial ideas have gone into designing good approximation algorithms for it. The state of approximability questions for its variants such as \emph{conductance} or \emph{edge expansion} is also similar.

Let us first define the \sparsestcut problem formally. The input is a pair of graphs $C$, $D$ on the same vertex set $V$, with $\size{V} = n$, called the \emph{cost} and \emph{demand} graphs, respectively. They are specified by non-negative edge weights  $c_{ij}, d_{ij} \geq 0$, for $i<j \in [n]$, and the \emph{(non-uniform) sparsest cut} problem, henceforth referred to as \sparsestcut, asks for a subset $S \subseteq V$ that minimizes
\[
\Phi(S) = \frac{\sum_{i<j} c_{ij} \abs{\indic_{S}(i) - \indic_{S}(j)}}{\sum_{i<j} d_{ij} \abs{\indic_{S}(i) - \indic_{S}(j)}},
\]
where $\indic_{S}(i)$ is the indicator function giving $1$, if $i \in S$, and $0$, otherwise. We denote the optimum by $\Phi^{*} = \min_{S \subseteq V} \Phi(S)$. The special case of this problem where the demand graph is a complete graph on $n$ vertices with uniform edge weights is called the \uniformSparsestcut problem.

Several popular heuristics in practice for finding sparse cuts use spectral information such as the eigenvalues and eigenvectors of the underlying graph. The \emph{generalized eigenvalues} of the Laplacian matrices of the cost and demand graphs, defined later in Section \ref{sec:prelim}, provide a natural scale against which we can measure the sparsity. If $0 \leq \lambda_{1} \leq \lambda_{2} \leq \dotsc \leq \lambda_{m}$ are the \emph{generalized eigenvalues} of the Laplacian matrices of cost and demand graphs, then using Courant-Fisher theorem (or the easy direction of Cheeger's inequality) we get $\lambda_{1} \leq \Phi^{*}$. So the smallest generalized eigenvalue is at most $\Phi^{*}$, and as we go to the higher eigenvalues, at some point they overtake $\Phi^{*}$. We provide an approximation guarantee of
\[
r~ \left(1 - \frac{\Phi^{*}}{\lambda_{r+1}}\right)^{-2}
\]
for the \sparsestcut problem, provided that $\lambda_{r+1} \geq \Phi^{*}$. In particular, this gives $O(r/\delta^{2})$ approximation guarantee, if $\lambda_{r+1} \geq \Phi^{*}/(1-\delta)$. Our algorithm runs in time $\text{poly}(n)$ and needs to solve a semidefinite program with only $\ell_{2}^{2}$ triangle inequality constraints. In comparison, Guruswami-Sinop \cite{GuruswamiS13} give an approximation guarantee of
\[
\left(1 - \frac{(1+\eps) \Phi^{*}}{\lambda_{r+1}}\right)^{-1},
\]
provided that $\lambda_{r+1} \geq (1+\eps) \Phi^{*}$, but require solving a semidefinite program with $O(r/\eps)$ level Lasserre constraints, and hence, $2^{r/\delta\eps} \text{poly}(n)$ running time \cite{GuruswamiS12b}.

\subsection{Our Results}
Our main result, proved later in Section \ref{Sec:SPCut}, is as follows:

\begin{theorem} \label{thm:MainThm}[Main Theorem]
Given an instance $C, D$ of the \sparsestcut problem, Algorithm \ref{sp-cut-alg} outputs a cut $T$ that satisfies
\[
\Phi(T) \leq \min_{r\in[n]} \quad r\left( 1 - \frac{\Phi^{*}}{\lambda_{r+1}}\right)^{-2} \Phi^{*}.
\]
The algorithm runs in time $\text{poly}(n)$ and needs to solve a semidefinite program with only additional $\ell_{2}^{2}$ triangle inequality constraints.
\end{theorem}

Here is an immediate corollary that was mentioned in the abstract.

\begin{coro} \label{cor:mainCor}
If the input instance satisfies  $\lambda_{r+1} \geq \Phi^{*}/(1-\delta)$ for some $r\in[n]$, then the algorithm produces a $O(r/\delta^2)$ approximation. Here, $0\leq \lambda_1 \leq \ldots \leq \lambda_n$ are the \emph{generalized eigenvalues} of the Laplacians of $C$, $D$.
\end{coro}

The proof of Theorem \ref{thm:MainThm} is based on the following property (see Subsection \ref{subsec:triangleineq}) of vectors in $\ell_2^2$ space that could be of independent interest.
\begin{prop} \label{prop:l2prop}
If $x_{1}, x_{2}, \dotsc, x_{n}$ satisfy $\ell_{2}^{2}$ triangle inequalities, then
\[
\inner{x_{i} - x_{j}}{\frac{x_{k} - x_{l}}{\norm{x_{k} - x_{l}}}}^{2} \leq \abs{\inner{x_{i} - x_{j}}{x_{k} - x_{l}}} \leq \norm{x_{i} - x_{j}}^{2}, \quad \text{for all $i, j, k, l \in [n]$}.
\]
\end{prop}
Geometrically, this gives an embedding of $x_{1}, x_{2}, \dotsc, x_{n}$ from $\ell_{2}^{2}$ into $\ell_{1}$ via appropriately scaled projections onto the line segment joining $x_{k}$ and $x_{l}$, for any $k \neq l$. Proposition \ref{prop:l2prop} says that this embedding is a contraction and the distortion for a pair is lower bounded by their squared distance after this projection. Thus, we can relate the average distortion to projections along certain directions.

\section{Previous and Related Work}
The \sparsestcut problem has seen a lot of activity, given its central importance. For the case of \uniformSparsestcut, where the demand graph is the complete graph with unit demands on all pairs, the first non-trivial bound was by using Cheeger's inequality (and a corresponding algorithm)\cite{AlonMilman85}. This gives an approximation factor of $1/\sqrt{\lambda_2(L)}$, where $\lambda_2(L)$ is the second-smallest eigenvalue of the normalized graph Laplacian matrix.

In in a seminal work, Leighton and Rao \cite{LeightonR99} related the problem of approximating the sparsest cut to embeddings between metric spaces, in particular, into $\ell_1$. By solving a LP relaxation of the \sparsestcut problem, they produce a metric on points and proceed to embed it into $\ell_1$, and show that the worst case distortion in doing so determines the approximation factor. Using a theorem of Bourgain, they obtain an $O(\log n)$ approximation.

Following this, the breakthrough work of Arora, Rao and Vazirani  \cite{AroraRV09} used a SDP (which we will refer to as ARV SDP) that could be viewed as a strengthening of both the spectral approach via Cheeger's inequality, and the distance metric approach of Leighton and Rao, to produce an $O(\sqrt{\log n})$ approximation for the \uniformSparsestcut. This SDP used the \emph{triangle inequality} constraints on the squared distances between vectors crucially, and was equivalent to the problem of embedding metrics from $\ell_2 ^2$ into $\ell_1$ with low \emph{average} distortion. Further work by Arora, Lee, and Naor  \cite{AroraLN05} extended these techniques to give an $O\left(\sqrt{\log n \log\log n}\right)$ approximation for the general \sparsestcut (equivalently, for the \emph{worst} case distortion of $\ell_2^2$ metrics into $\ell_1$).

Recently, Guruswami and Sinop \cite{GuruswamiS13-general} gave a generic method for rounding a class of SDP hierarchies proposed by Lasserre \cite{Lasserre01, Laurent03}, and applied it to the \sparsestcut problem \cite{GuruswamiS13}. This hierarchy subsumes the ARV SDP within $3$-levels, but the size of their SDP with $r$ levels increases as $n^{O(r)}$. The approximation guarantee depends on the \emph{generalized eigenvalues} of the pair of Laplacians of the cost and demand graphs, and is as follows:

\begin{theorem}[Guruswami-Sinop \cite{GuruswamiS13}]
Given $C,D$ as cost and demand graphs let $0\leq \lambda_1 \leq \lambda_2 \ldots \leq \lambda_n$ be the generalized eigenvalues between $C,D$. Then for every $r \in [n]$ and $\epsilon \geq 0$, a solution satisfying $O(r/\epsilon)$ levels of the Lasserre hierarchy with objective value $\Phi^*$ can be rounded to produce a cut $T$ with value
\[
\Phi(T) \leq \Phi^* \left(1 - \frac{(1+\eps) \Phi^*}{\lambda_{r+1}}\right)^{-1}, \quad   \text{if $\lambda_{r+1} \geq (1+\eps) \Phi^{*}$}.
\]
\end{theorem}

For the specific case of the \uniformSparsestcut problem, Arora, Ge and Sinop \cite{AroraGS13} show, by using techniques from  Guruswami-Sinop, that under certain conditions on the input graph (expansions of sets of size $\leq n/r$), they can get a $(1+\epsilon)$ approximation; again using the $r$-th level of the Lasserre hierarchy.

On the side of integrality gaps, the best known integrality gap for the (non-uniform) ARV SDP is $(\log n)^{\Omega(1)}$ by Cheeger, Kleiner and Naor \cite{CheegerKN09}.

The main motivation behind this work is to get approximation guarantees similar to the Guruswami-Sinop rounding \cite{GuruswamiS13}, but without using higher levels of the Lasserre hierarchy. Some parts of the Guruswami-Sinop proof such as column subset selection via volume sampling do not require higher level Lasserre vectors or constraints. Also the final approximation guarantee of Guruswami-Sinop does not depend on higher level Lasserre vectors. While our approximation guarantee is mildly worse than theirs, our algorithm always runs in polynomial time and does not use higher level Lasserre vectors in the rounding.

\section{Notation and Preliminaries} \label{sec:prelim}

We state the necessary notation and definitions formally in this section.

\paragraph*{Sets, Matrices, Vectors}
We use $[n]={1,\ldots,n}$. For a matrix $X \in \mathbb{R}^{d\times d}$, we say $X\succeq 0$ or $X$ is positive-semidefinite if $y^TXy \geq 0$ for all $y\in \mathbb{R}^d$. The Gram-matrix of a matrix $M \in \mathbb{R}^{d_1 \times d_2}$ is the matrix $M^T M$, which is positive-semidefinite. We will often need the eigenvalues of the Gram-matrix of $M$. We will denote these by $\sigma_1(M)\geq \sigma_2(M) \geq \ldots \sigma_{d_2}(M)\geq 0$, arranged in descending order. The \emph{Frobenius} norm of $M$ is given by $\fnorm{M} \triangleq \sqrt{\sum_i \sigma_i(M)} = \sqrt{\sum_{i\in[d_1],j\in[d_2]}M(i,j)^2}$. In our analysis, we will sometimes view a matrix $M$ as a collection of its columns viewed as vectors; $M=(m_j)_{j\in[d_2]}$. In this case, $\fnorm{M}^2=\sum_j \norm{m_j}^2$.

\paragraph*{Generalized Eigenvalues}
Given two symmetric matrices $X, Y\in \mathbb{R}^d\times d$ with $Y\succeq 0$, and for
$i \leq \text{rank}(Y)$, we define their $i$-th smallest generalized eigenvalue as the following:
\[
 \lambda_i = \max_{\text{rank}(Z)\leq i-1 \quad} \min_{w \bot Z; w\neq 0 \quad} \frac{w^TXw}{w^TYw}
\]

\paragraph*{Graphs and Laplacians}
All graphs will be defined on a vertex set $V$ of size $n$. The vertices will usually be referred to by indices $i,j,k,l \in[n]$. Given a graph with weights on pairs $W:{V \choose 2}\mapsto \mathbb{R}^+$, the graph Laplacian matrix is defined as:
\begin{align*}
 L_W(i,j) =
 \begin{cases}
 -W(i,j) &\quad \text{if $i\neq j$} \\
 \sum_k W(i,k) &\quad \text{if $i=j$} \\
 \end{cases}
\end{align*}

\paragraph*{\sparsestcut SDP}\label{sdp:ARVSDP}
The SDP we use for \sparsestcut on the vertex set $V$ with costs and demands $c_{ij}, d_{kl} \geq 0$  and corresponding cost and demand graphs $C:{V \choose 2}\mapsto \mathbb{R}^+$ and $D:{V\choose 2}\mapsto \mathbb{R}^+$, is effectively the following:

\begin{align}
\textbf{SDP:}\quad \Phi(SDP) & = \min \frac{\sum_{i<j} c_{ij} \norm{x_{i}-x_{j}}^2}{\sum_{k<l} d_{kl} \norm{x_{k}-x_{l}}^2} \\
&\text{subject to}\quad  \norm{x_{i}-x_{j}}^2 +\norm{x_{j}-x_{k}}^2\geq \norm{x_{i}-x_{k}}^{2} & \forall i,j,k\in[n]
\end{align}

While this is technically not an SDP due to the presence of a fraction in the objective, it is not difficult to see that we can construct an equivalent SDP as shown in \cite{GuruswamiS13}.
We will use $\Phi(ALG)$ to denote the sparsity of the cut produced by an algorithm, and will compare it to $ \Phi(SDP)$. Note that any set of vectors ${x_1, \ldots, x_n}$ that are feasible for this SDP satisfy the triangle inequalities on the \emph{squares} of their distances, and are said to satisfy the $\ell_2^2$ triangle inequality, or are in $ \ell_2 ^2 $ space.

\paragraph*{Lasserre Hierarchy}
The Lasserre hierarchy \cite{Lasserre01} at level $r$ strengthens the basic SDP relaxation by introducing new vectors, $x_S(f)$, for every $S\subseteq[n]$ with $|S|\leq r$ and every $f:S\rightarrow \{0,1\}^{|S|}$, and requiring certain consistency conditions on the inner products between them. We do not go into the details of the hierarchy here, since we will not be using it in this work. We refer the reader to available surveys, e.g. \cite{Laurent03} for more details. For the \sparsestcut problem, one can show that the $\ell_{2}^{2}$ triangle inequalities are subsumed by $3$ levels of this hierarchy.

\paragraph*{$\ell_1$ embeddings and cuts}
Leighton and Rao \cite{LeightonR99} show that instead of producing cuts, it is sufficient to produce a mapping $Z: V\rightarrow \mathbb{R}^d$, with $z_i= Z(i)$, from which we can extract a cut $T$ such that
\[ \Phi(T) \leq \frac{\sum_{i<j}c_{ij}\Vert z_i - z_j \Vert_1 }{\sum_{k<l}d_{kl}\Vert z_k - z_l \Vert_1 } \].

This follows from the fact that  $\ell_1$ metrics are exactly the cone of cut-metrics.

\section{Lasserre hierarchy vs. $\ell_{2}^{2}$ triangle inequality}
Let's first recap Guruswami-Sinop \cite{GuruswamiS13,GuruswamiS13-general,GuruswamiS12a,GuruswamiS12b} to demonstrate its key ideas and to facilitate its comparison with our method coming later. At the basic level, they map SDP solution vectors to values in $[0, 1]$, where one can then run independent or threshold rounding. To define this map, they need $O(r)$-level Lasserre vectors $\{x_{S}(f)\}_{S,f}$ for subsets $S \subseteq [n]$ of size at most $O(r)$ and assignments $f \in \{0, 1\}^{\size{S}}$. For simplicity of notation, call $x_{\{i\}}(1)$ as $x_{i}$. Now the algorithm has two parts.
\begin{enumerate}
\item Pick a subset $S$ of size $O(r)$ using volume sampling \cite{GuruswamiS11} on the matrix with columns as $\{\sqrt{d_{ij}}(x_{i} - x_{j})\}_{i<j}$. This part does not require Lasserre vectors or constraints in the algorithm as well as the analysis.
\item For the $S$ fixed as above, pick $x_{S}(f)$ with probability $\propto \norm{x_{S}(f)}^{2}$ and map each $x_{i}$ to $p_{i}^{(f)} \in [0, 1]$ as follows.
\[
x_{i} \mapsto p_{i}^{(f)} = \frac{\inner{x_{i}}{x_{S}(f)}}{\norm{x_{S}(f)}^{2}} \in [0, 1].
\]
Once we have $p_{i}^{(f)} \in [0, 1]$ for all $i \in [n]$, we can either do threshold rounding with a random threshold $r \in [0, 1]$ or do independent rounding with $p_{i}^{(f)}$'s as probabilities. Lasserre constraints are used to show $p_{i}^{(f)} \in [0, 1]$ and the following important property used in the analysis.
\[
\inner{x_{i} - x_{j}}{\frac{x_{S}(f)}{\norm{x_{S}(f)}}}^{2} \leq \abs{\inner{x_{i} - x_{j}}{x_{S}(f)}} \leq \norm{x_{i} - x_{j}}^{2}, \quad \text{for all $i, j \in [n]$}.
\]
\end{enumerate}
What is special about these directions $x_{S}(f)$? Are there other directions that exhibit similar property and can be found without solving multiple levels of Lasserre hierarchy?

\subsection{$\ell_{2}^{2}$ triangle inequality} \label{subsec:triangleineq}
We make an interesting observation that $\ell_{2}^{2}$ triangle inequalities give a large collection of vectors that exhibit the same property as the $x_{S}(f)$'s used in the analysis of Guruswami-Sinop. $\ell_{2}^{2}$ triangle inequalities for all triplets, or equivalently, the acuteness of all angles in a point set $\{x_{1}, x_{2}, \dotsc, x_{n}\}$ can be written as $\inner{x_{i} - x_{l}}{x_{k} - x_{l}} \geq 0$, for all $i, k, l \in [n]$, and gives the following interesting mapping of vectors $x_{i}$ to values $p_{i}^{(k, l)} \in [0, 1]$ as
\[
x_{i} \mapsto p_{i}^{(k, l)} = \frac{\inner{x_{i} - x_{l}}{x_{k} - x_{l}}}{\norm{x_{k} - x_{l}}^{2}}.
\]
Note that $p_{i}^{(k, l)}$ depends on the ordered pair $(k, l)$, and $p_{i}^{(k, l)} \in [0, 1]$ by the $\ell_{2}^{2}$ triangle inequalities or acuteness of all angles. Another interesting consequence is
\[
1 - p_{i}^{(k, l)} = \frac{\inner{x_{k} - x_{i}}{x_{k} - x_{l}}}{\norm{x_{k} - x_{l}}^{2}}.
\]
Moreover, we show that the direction $x_{k} - x_{l}$ behaves similar to $x_{S}(f)$ used in the analysis of Guruswami-Sinop.
\begin{prop} \label{prop:triangle} [Restatement of Proposition \ref{prop:l2prop}]
If $x_{1}, x_{2}, \dotsc, x_{n}$ satisfy $\ell_{2}^{2}$ triangle inequalities, then
\[
\inner{x_{i} - x_{j}}{\frac{x_{k} - x_{l}}{\norm{x_{k} - x_{l}}}}^{2} \leq \abs{\inner{x_{i} - x_{j}}{x_{k} - x_{l}}} \leq \norm{x_{i} - x_{j}}^{2}, \quad \text{for all $i, j, k, l \in [n]$}.
\]
\end{prop}
\begin{proof}
By acuteness of all angles, we know that
\[
\inner{x_{i} - x_{k}}{x_{i} - x_{j}} \geq 0 \quad \text{and} \quad \inner{x_{l} - x_{j}}{x_{i} - x_{j}} \geq 0, \quad \text{for all $i, j, k, l \in [n]$}.
\]
Adding both the inequalities we get $\norm{x_{i} - x_{j}}^{2} - \inner{x_{k} - x_{l}}{x_{i} - x_{j}} \geq 0$, or equivalently $\inner{x_{k} - x_{l}}{x_{i} - x_{j}} \leq \norm{x_{i} - x_{j}}^{2}$. Since swapping $k$ and $l$ does not affect the above argument, we get the upper bound
\[
\abs{\inner{x_{k} - x_{l}}{x_{i} - x_{j}}} \leq \norm{x_{i} - x_{j}}^{2}, \quad \text{for all $i, j, k, l \in [n]$}.
\]
Swapping $(i, j)$ and $(k, l)$, we also have $\abs{\inner{x_{k} - x_{l}}{x_{i} - x_{j}}} \leq \norm{x_{k} - x_{l}}^{2}$. Therefore,
\begin{align*}
\inner{x_{i} - x_{j}}{\frac{x_{k} - x_{l}}{\norm{x_{k} - x_{l}}}}^{2} & = \frac{\inner{x_{k} - x_{l}}{x_{i} - x_{j}}^{2}}{\norm{x_{k} - x_{l}}^{2}} \\
& \leq  \frac{\inner{x_{k} - x_{l}}{x_{i} - x_{j}}^{2}}{\abs{\inner{x_{k} - x_{l}}{x_{i} - x_{j}}}} \\
& = \abs{\inner{x_{k} - x_{l}}{x_{i} - x_{j}}}.
\end{align*}
\end{proof}

\subsection{Low dimensional SDP solutions}
Although the Guruswami-Sinop \cite{GuruswamiS13} result is finally stated in terms of a condition on generalized eigenvalues, it can also be thought of as a result that gives good approximation guarantees when the SDP solution is close to being low rank. Suppose the Gram matrix of $\{x_{i} - x_{j}\}_{1 \leq i < j \leq n}$ has at least $\delta$ fraction of its spectrum in its top $r$ eigenvalues, that is, $\sum_{t=1}^{r} \lambda_{t} \geq \delta~ \sum_{t=1}^{n} \lambda_{t}$, where $\lambda_{1} \geq \lambda_{2} \geq \dotsc \geq \lambda_{n} \geq 0$ are the eigenvalues of the Gram matrix of $\{x_{i} - x_{j}\}_{1 \leq i < j \leq n}$. Then Proposition \ref{prop:spectral} proves the existence of a good direction $x_{k} - x_{l}$ by weighted averaging.

\begin{prop} \label{prop:spectral}
If $x_{1}, x_{2}, \dotsc, x_{n}$ satisfy the above spectral or low-rank property, then there exists $x_{k} - x_{l}$ such that
\[
\sum_{i<j} \inner{x_{i} - x_{j}}{\frac{x_{k} - x_{l}}{\norm{x_{k} - x_{l}}}}^{2} \geq \frac{\delta^{2}}{r} \sum_{i<j} \norm{x_{i} - x_{j}}^{2}.
\]
\end{prop}
\begin{proof}
To show the existence of a good $x_{k} - x_{l}$, we take expectation over $x_{k} - x_{l}$ by squared length sampling.
\begin{align*}
\max_{k<l}~ \sum_{i<j} \inner{x_{i} - x_{j}}{\frac{x_{k} - x_{l}}{\norm{x_{k} - x_{l}}}}^{2} & \geq \sum_{k<l} \frac{\norm{x_{k} - x_{l}}^{2}}{\sum_{p<q} \norm{x_{p} - x_{q}}^{2}} \sum_{i<j} \inner{x_{i} - x_{j}}{\frac{x_{k} - x_{l}}{\norm{x_{k} - x_{l}}}}^{2} \\
& = \frac{\sum_{k<l} \sum_{i<j} \inner{x_{i} - x_{j}}{x_{k} - x_{l}}^{2}}{\sum_{p<q} \norm{x_{p} - x_{q}}^{2}} \\
& = \frac{\sum_{t=1}^{n} \lambda_{t}^{2}}{\sum_{t=1}^{n} \lambda_{t}} \\
& \geq \frac{\sum_{t=1}^{r} \lambda_{t}^{2}}{\sum_{t=1}^{n} \lambda_{t}} \\
& \geq \frac{\left(\sum_{t=1}^{r} \lambda_{t}\right)^{2}}{r~ \sum_{t=1}^{n} \lambda_{t}} \qquad \text{by Cauchy-Schwarz inequality} \\
& \geq \frac{\delta^{2}}{r}~ \sum_{t=1}^{n} \lambda_{t} \qquad \text{by the spectral or low-rank property} \\
& = \frac{\delta^{2}}{r}~ \sum_{i<j} \norm{x_{i} - x_{j}}^{2}.
\end{align*}
\end{proof}

\section{Non-uniform sparsest cut}\label{Sec:SPCut}

We now give the proof of the Main Theorem (Theorem \ref{thm:MainThm}). The rounding algorithm is Algorithm \ref{sp-cut-alg}.

\begin{algorithm}
\caption{ Algorithm for \sparsestcut} \label{sp-cut-alg}
\begin{algorithmic}[1]
\REQUIRE $C,D$ and a solution $\{x_1, \ldots, x_n\}$ to the ARV SDP for \sparsestcut
\ENSURE A cut $(T, \bar{T})$
\FORALL{Pairs $(k,l)\in [n]\times[n]$}
\STATE $p_i^{(k,l)} = \dfrac{\inner{x_{i} - x_{l}}{x_{k} - x_{l}}}{\norm{x_{k} - x_{l}}^{2}}$ \qquad {\tt \% line embedding}
\FORALL{$t\in[n]$}
\STATE $S_{kl}^{(t)} = \left\{i \suchthat p_i^{(k,l)} \leq p_t^{(k,l)}\right\}$ \qquad {\tt \% threshold rounding}
\ENDFOR
\ENDFOR
\STATE $T=\arg\min_{k,l,t} \Phi\left(S_{kl}^{(t)}\right)$
\STATE Output the cut $(T,\bar{T})$
\end{algorithmic}
\end{algorithm}

Algorithm \ref{sp-cut-alg} goes over all directions $x_{k} - x_{l}$. For each of them, it maps $x_{i}$ to $p_{i} \in [0, 1]$ as
\[
x_{i} \mapsto p_{i}^{(k, l)} = \frac{\inner{x_{i} - x_{l}}{x_{k} - x_{l}}}{\norm{x_{k} - x_{l}}^{2}}.
\]
Now for each $t \in [n]$ consider the sweep cut $S_{t} = \{j \suchthat p_{j}^{(k, l)} \leq p_{t}^{(k, l)}\}$, and output the best amongst them as $T$.

For convenience of notation, we will do the analysis using the corresponding $\ell_{1}$ embedding, as mentioned in Section \ref{sec:prelim}. Given an $\ell_{1}$-embedding, we can get a cut with similar guarantee by choosing the best threshold cut along each coordinate, which is what our algorithm does. Define an $\ell_{1}$-embedding of $x_{i}$'s as follows.
\[
x_{i} \mapsto y_{i} = \left(\frac{d_{kl} \norm{x_{k} - x_{l}}^{2} \inner{x_{i} - x_{l}}{x_{k} - x_{l}}}{\sum_{k<l} d_{kl} \norm{x_{k} - x_{l}}^{2}}\right)_{k<l}.
\]
The following is an easy consequence of Proposition \ref{prop:triangle}.
\begin{prop} \label{prop:distortion}
\[
\frac{\sum_{k<l} d_{kl} \inner{x_{i} - x_{j}}{x_{k} - x_{l}}^{2}}{\sum_{k<l} d_{kl} \norm{x_{k} - x_{l}}^{2}} \leq \norm{y_{i} - y_{j}}_{1} \leq \norm{x_{i} - x_{j}}^{2},
\]
\end{prop}
\begin{proof}
Let's start with the upper bound.
\begin{align*}
\norm{y_{i} - y_{j}}_{1} & = \frac{\sum_{k<l} d_{kl} \norm{x_{k} - x_{l}}^{2} \abs{\inner{x_{i} - x_{j}}{x_{k} - x_{l}}}}{\sum_{k<l} d_{kl} \norm{x_{k} - x_{l}}^{2}} \\
& \leq \frac{\sum_{k<l} d_{kl} \norm{x_{k} - x_{l}}^{2} \norm{x_{i} - x_{j}}^{2}}{\sum_{k<l} d_{kl} \norm{x_{k} - x_{l}}^{2}} \qquad \text{by Proposition \ref{prop:triangle}} \\
& = \norm{x_{i} - x_{j}}^{2}.
\end{align*}
Now the lower bound.
\begin{align*}
\norm{y_{i} - y_{j}}_{1} & = \frac{\sum_{k<l} d_{kl} \norm{x_{k} - x_{l}}^{2} \abs{\inner{x_{i} - x_{j}}{x_{k} - x_{l}}}}{\sum_{k<l} d_{kl} \norm{x_{k} - x_{l}}^{2}} \\
& \geq \frac{\sum_{k<l} d_{kl} \norm{x_{k} - x_{l}}^{2} \inner{x_{i} - x_{j}}{\frac{x_{k} - x_{l}}{\norm{x_{k} - x_{l}}}}^{2}}{\sum_{k<l} d_{kl} \norm{x_{k} - x_{l}}^{2}} \qquad \text{by Proposition \ref{prop:triangle}} \\
& = \frac{\sum_{k<l} d_{kl} \inner{x_{i} - x_{j}}{x_{k} - x_{l}}^{2}}{\sum_{k<l} d_{kl} \norm{x_{k} - x_{l}}^{2}}.
\end{align*}
\end{proof}
Equipped with this, we can now bound the average distortion, and hence, the approximation factor of our algorithm. We use the following Proposition from Guruswami-Sinop \cite{GuruswamiS13} to rewrite the final bound in terms of the generalized eigenvalues of the Laplacian matrices of the cost and demand graphs.
\begin{prop} \label{prop:laplacian} \cite{GuruswamiS13}
Let $0 \leq \lambda_{1} \leq \dotsc \leq \lambda_{m}$ be the generalized eigenvalues of the Laplacian matrices of the cost and demand graphs. Let $\sigma_{1} \geq \sigma_{2} \geq \dotsc \geq \sigma_{n} \geq 0$ be eigenvalues of the Gram matrix of vectors $\{\sqrt{d_{ij}} (x_{i} - x_{j})\}_{i<j}$. Then
\[
\frac{\sum_{t \geq r+1} \sigma_{j}}{\sum_{t=1}^{n} \sigma_{j}} \leq \frac{\Phi(SDP)}{\lambda_{r+1}}.
\]
\end{prop}
Using these we bound the approximation ratio of our algorithm and prove Theorem \ref{thm:MainThm}.
\begin{theorem} \label{thm:nonunif} [Restatement of Theorem \ref{thm:MainThm}]
\[
\Phi(ALG) \leq \Phi(SDP) \cdot r~ \left(1 -\frac{\Phi(SDP)}{\lambda_{r+1}}\right)^{-2}.
\]
\end{theorem}
\begin{proof}
The guarantee of our algorithm can only be better than the guarantee of this corresponding $\ell_{1}$-embedding.
\begin{align*}
\Phi(ALG) & \leq \frac{\sum_{i<j} c_{ij} \norm{y_{i} - y_{j}}_{1}}{\sum_{i<j} d_{ij} \norm{y_{i} - y_{j}}_{1}} \\
& \leq \frac{\sum_{i<j} c_{ij} \norm{x_{i} - x_{j}}^{2}~ \sum_{k<l} d_{kl} \norm{x_{k} - x_{l}}^{2}}{\sum_{i<j} d_{ij} \sum_{k<l} d_{kl} \inner{x_{i} - x_{j}}{x_{k} - x_{l}}^{2}} \\
& = \frac{\sum_{i<j} c_{ij} \norm{x_{i} - x_{j}}^{2}}{\sum_{i<j} d_{ij} \norm{x_{i} - x_{j}}^{2}} \cdot \frac{\left(\sum_{i<j} d_{ij} \norm{x_{i} - x_{j}}^{2}\right) \left(\sum_{k<l} d_{kl} \norm{x_{k} - x_{l}}^{2}\right)}{\sum_{i<j} \sum_{k<l} d_{ij} d_{kl} \inner{x_{i} - x_{j}}{x_{k} - x_{l}}^{2}} \\
& = \Phi(SDP) \cdot \frac{\left(\sum_{i<j} d_{ij} \norm{x_{i} - x_{j}}^{2}\right)^{2}}{\sum_{i<j} \sum_{k<l} d_{ij} d_{kl} \inner{x_{i} - x_{j}}{x_{k} - x_{l}}^{2}} \\
& = \Phi(SDP) \cdot \frac{\left(\sum_{t=1}^{n} \sigma_{t}\right)^{2}}{\sum_{t=1}^{n} \sigma_{t}^{2}} \\
& \leq \Phi(SDP) \cdot \frac{\left(\sum_{t=1}^{n} \sigma_{t}\right)^{2}}{\sum_{t=1}^{r} \sigma_{t}^{2}} \\
& \leq \Phi(SDP) \cdot r~ \left(\frac{\sum_{t=1}^{n} \sigma_{t}}{\sum_{t=1}^{r} \sigma_{t}}\right)^{2} \qquad \text{by Cauchy-Schwarz inequality} \\
& \leq \Phi(SDP) \cdot r~ \left(1 - \frac{\sum_{t \geq r+1} \sigma_{t}}{\sum_{t=1}^{n} \sigma_{t}}\right)^{-2} \\
& \leq \Phi(SDP) \cdot r~ \left(1 - \frac{\Phi(SDP)}{\lambda_{r+1}}\right)^{-2} \qquad \text{by Proposition \ref{prop:laplacian}} \\
& \leq \Phi^{*} \cdot r~ \left(1 - \frac{\Phi^{*}}{\lambda_{r+1}}\right)^{-2}.
\end{align*}
\end{proof}

\section{Conclusion}
We show that it is possible to get approximation guarantees similar to Guruswami-Sinop for the \sparsestcut problem, but without using higher level Lasserre vectors. One obvious question that arises out of this is whether we can apply these techniques with threshold or independent rounding to give similar guarantees for other problems. Further, can we obtain more directions for projections and sweep cuts using lower levels of the Lasserre hierarchy or eigenvectors of the SDP solution?

\subsection*{Acknowledgement}
The authors would like to thank Prahladh Harsha for many valuable discussions.

\bibliography{ref}

\nocite{*}

\end{document}